\documentclass[11pt]{article}

\usepackage{fullpage}
\usepackage{times}
\usepackage{amssymb, amsmath}
\usepackage{color}
\usepackage{graphicx}
\usepackage{verbatim} 
\usepackage{url}

\newtheorem{theorem}{Theorem}
\newtheorem{lemma}{Lemma}

\newtheorem{corollary}{Corollary}

\newtheorem{remark}{Remark}
\newtheorem{conjecture}{Conjecture}

\newcommand{\eps}{\varepsilon}

\newcommand{\E}{\mathsf{E}}

\renewcommand{\Pr}{\mathsf{Pr}}

\newcommand{\abs}[1]{\left| #1 \right|}

\newcommand{\OR}{\mathrm{OR}}

\newcommand{\DISJ}{{$2$-DISJ$^r$}}
\newcommand{\DISJP}{{$2$-DISJ$^h$}}
\newcommand{\DISJE}{{$2$-DISJ$^{m}$}}

\newcommand{\orDISJ}{{OR-DISJ$^r$}}

\newcommand{\thresh}{{THRESH$_\theta^r$}}
\newcommand{\threshR}{{THRESH$_{(3r-1)/4}^r$}}
\newcommand{\threshP}{{THRESH$_{(3n-1)/4}^n$}}
\newcommand{\threshM}{{THRESH$_{(3r-1)/4}^r$ where $r = \min\{m/k, n\}$}}
\newcommand{\threshE}{{THRESH$_{(3m-1)/4}^m$}}

\newcommand{\qedsymb}{\hfill{\rule{2mm}{2mm}}}
\def\FullBox{\hbox{\vrule width 8pt height 8pt depth 0pt}}
\def\qed{\ifmmode\qquad\FullBox\else{\unskip\nobreak\hfil
\penalty50\hskip1em\null\nobreak\hfil\FullBox
\parfillskip=0pt\finalhyphendemerits=0\endgraf}\fi}

\newenvironment{proof}{\begin{trivlist}
\item[\hspace{\labelsep}{\bf\noindent Proof: }]
}{\qedsymb\end{trivlist}}

\begin{document}
\title{When Distributed Computation is Communication Expensive}
\author{David P. Woodruff\\IBM Research Almaden
%\\650 Harry Road\\San Jose, CA 95136
\\dpwoodru@us.ibm.com
\and Qin Zhang~\footnote{Most of this work was done when the author was at IBM Research Almaden.}\\Indiana University Bloomington\\qinzhang@cse.ust.hk}

\date{}

%\begin{titlepage}
\maketitle

\begin{abstract}
We consider a number of fundamental statistical and graph problems in the message-passing model, where we have $k$ machines (sites), each holding a piece of data, and the machines want to jointly solve a problem defined on the union of the $k$ data sets. The communication is point-to-point, and the goal is to minimize the total communication among the $k$ machines. This model captures all point-to-point distributed computational models with respect to minimizing communication costs. Our analysis shows that exact computation of many statistical and graph problems in this distributed setting requires a prohibitively large amount of communication, and often one cannot improve upon the communication of the simple protocol in which all machines send their data to a centralized server. Thus, in order to obtain protocols that are communication-efficient, one has to allow approximation, or investigate the distribution or layout of the data sets.
\end{abstract}
%\bigskip

%This submission is a regular paper submission.
%\end{titlepage}

\section{Introduction}
\label{sec:intro}
Recent years have witnessed a spectacular increase in the amount of data being collected and processed in various applications. In many of these applications, data is often distributed across a group of machines, referred to as {\em sites} in this paper, which are connected by a communication network. These sites jointly compute a function defined on the union of the data sets by exchanging messages with each other. For example, consider the following scenarios.
\begin{enumerate}
\item
We have a set of network routers, each observing a portion of the network, and periodically they want to compute some functions defined on the global network which can be used to determine the overall condition/health of the network. Concrete functions include the number of distinct source IP addresses, the set of most frequent destination IP addresses, etc.

\item The massive social network graphs are usually stored in many sites, and those graphs are keeping changing.  To answer queries such as whether the whole graph is connected, or whether the graph exhibit a particular property (e.g., bipartiteness, cycle-freeness), we have to synthesize data from all the sites.
\end{enumerate}
%Concrete models for this type of computation include the BSP model~\cite{Valiant90} and the recent and extensively-studied MapReduce model~\cite{dean08:_mapred}. Popular system architectures include Hadoop \cite{hadoop}, Google's Pregel \cite{pregel}, Microsoft's Trinity \cite{trinity}.
%, as well as the recent PowerGraph \cite{GLGBG12}. 
%Horton \cite{SEYG12} and 

In distributed computational models for big data, besides traditional measurement like local CPU processing time and the number of disk accesses, we are also interested in minimizing two other objectives, namely, the {\em communication cost} and the {\em round complexity}. The communication cost, which we shall also refer to as the communication complexity, denotes the total number of bits exchanged in all messages across the sites during a computation. The round complexity refers to the number of communication rounds needed for the computation, given various constraints on what messages can be sent by each site in each round.

%\qinsays{what do you mean by ``the number of messages, which we shall also refer to as rounds, exchanged in order to complete the computation, without specifying how many bits need to be exchanged". This is NOT the definition of rounds in the distributed setting. Have modified.}

The communication cost is a fundamental measure since communication is often the bottleneck of applications (e.g., applications mentioned above), and so it directly relates to energy consumption, network bandwidth usage, and overall running time. The round complexity is critical when the computation is partitioned into rounds and the initialization of each round requires a large overhead. In this paper we will focus on the communication complexity, and analyze problems in an abstract model called the message-passing model (see the definition in Section~\ref{sec:model}) that captures all models for point-to-point distributed computation in terms of their communication costs. In particular, our lower bound results hold even if the communication protocol sends only a single bit in each message, and each site has an unbounded amount of local memory and computational power. Note that this means our lower bounds are as strong as possible, not requiring any assumptions on the local computational power of the machines. We also present several upper bounds, all of which are also locally computationally efficient, meaning the protocols we present do not need extra memory beyond what is required to accommodate the input. We will briefly discuss the issue of round-efficiency in Section~\ref{sec:discussion}.

Common sources of massive data include numerical data, e.g., IP streams and logs of queries to a search engine, as well as graph data, e.g., web graphs, social networks, and citation graphs. In this paper we investigate the communication costs for solving several basic statistical and graph problems in the message-passing model. Solving these problems is a minimal requirement of protocols seeking to solve more complicated functions on distributed data. 

We show that if we want to solve many of these problems exactly, then there are no better solutions than the almost trivial ones, which are usually quite communication-inefficient. The motivation of this work is thus to deliver the following message to people working on designing protocols for solving problems on distributed systems: for many statistical and graph problems in the distributed setting, if we want efficient communication protocols, then we need to consider the following relaxations to the original problem:
\begin{enumerate}
\item Allow for returning an approximate solution. Here, approximation can be defined as follows: for a problem whose output is a single numerical value $x$, allowing an approximation means that the protocol is allowed to return any value $\tilde{x}$ for which $\tilde{x} \in [(1-\eps)x, (1+\eps)x]$, for some small user-specified parameter $\eps > 0$. For a problem whose output is YES or NO, e.g., a problem deciding if a certain property of the input exists or not, we could instead allow the protocol to return YES if the input is close to having the property (under some problem-specific notion of closeness) and NO if the input is far from having that property. For example, in the graph connectivity problem, we return YES if the graph can be made connected by adding a small number of edges, while we return NO if the graph requires adding a large number of edges to be made connected. This latter notion of approximation coincides with the {\em property testing} paradigm \cite{GGR98} in the computer science literature.  

By allowing certain approximations we can sometimes drastically reduce the communication costs. Concrete examples and case studies will be given in Section~\ref{sec:F0} and Section~\ref{sec:diameter}.

 % For example, we will show in this paper that computing the number of distinct elements exactly in the message-passing model with $k$ sites needs $\Omega(k \cdot F_0)$ bits, where $F_0$ denotes the number of distinct elements, while in \cite{CMY11} it was shown that $F_0$ can be approximated by a $(1+\eps)$ multiplicative factor using $O(k\log n/\eps^2)$ bits of communication, where $n$ is the universe size. This improves the exact computation by almost a factor of $F_0$! 

\item Use well-designed input layouts. Here are two examples: (1) All edges from the same node are stored in the same site or on only a few sites. In our lower bounds the edges adjacent to a node are typically stored across many different sites. (2) Each edge is stored on a unique (or a small number) of different sites. Our results in Table~\ref{tab:results} show that whether or not the input graph has edges that occur on multiple sites can make a huge difference in the communication costs. 

\item Explore prior distributional properties of the input dataset, e.g., if the dataset is skewed, or the underlying graph is sparse or follows a power-law distribution. Instead of developing algorithms targeting the worst-case distributions, as those used in our lower bounds, if one is fortunate enough to have a reasonable model of the underlying distribution of inputs, this can considerably reduce communication costs. An extreme example is that of a graph on $n$ vertices - if the graph is completely random, meaning, each possible edge appears independently with probability $p$, then the $k$ sites can simply compute the total number of edges $m$ to decide whether or not the input graph is connected with high probability. Indeed, by results in random graph theory, if $m \ge 0.51n\log n$ then the graph is connected with very high probability, while if $m \le 0.49n\log n$ then the graph is disconnected with very high probability~\cite{ER60}. Of course, completely random graphs are unlikely to appear in practice, though other distributional assumptions may also result in more tractable problems. 
\end{enumerate}
%\qinsays{I put the last item since all the LBs proved in this paper target worse distributions. I think a good point about LB is that it tells us which (at least one) distributions are hard, and then we can see whether the distributions in the reality is indeed hard. I agree that this point is not very strong, but I always like to talk about it. Let me know what you think.}

%\vfill\eject
\subsection{The Message-Passing model}
\label{sec:model}
In this paper we consider the message-passing model, studied, for example, in \cite{PVZ12,WZ12}. In this model we have $k$ sites, e.g., machines, sensors, database servers, etc., which we denote as $P_1, \ldots, P_k$. Each site has some portion of the overall data set, and the sites would like to compute a function defined on the union of the $k$ data sets by exchanging messages. There is a two-way communication channel between all pairs of players $P_i$ and $P_j$. Then, since we will prove lower bounds, our lower bounds also hold for topologies in which each player can only talk to a subset of other players. The communication is point-to-point, that is, if $P_i$ talks to $P_j$, then the other $k - 2$ sites do not see the messages exchanged between $P_i$ and $P_j$. At the end of the computation, at least one of the sites should report the correct answer. The goal is to minimize the total number of bits and messages exchanged among the $k$ sites. For the purposes of proving impossibility results, i.e., lower bounds, we can allow each site to have an infinite local memory and infinite computational power; note that such an assumption will only make our lower bounds stronger. Further, we do not place any constraints on the format of messages or any ordering requirement on the communication, as long as it is point-to-point. 

The message-passing model captures all point-to-point distributed communication models in terms of the communication cost, including the BSP model by Valiant~\cite{Valiant90}, the $\mathcal{M}\mathcal{R}\mathcal{C}$ MapReduce framework proposed by Karloff et al.~\cite{KSV10}, the generic MapReduce model by Goodrich et al.~\cite{GSZ11}, and the Massively Parallel model by Koutris and Suciu~\cite{KS11}.

\begin{remark}
We comment that in some settings, where the primary goal is to parallelize a single computation on a big dataset, communication may not be the only bottleneck; CPU time and disk accesses are also important. %Sometimes, the cost to read the whole data from the source to $k$ sites could dominate the later communication cost. 
However, in this paper we are mainly interested in the following setting: The data has already been distributed to the sites, and perhaps keeps changing. The goal is to periodically compute some function that is defined on the dataset (e.g., queries). In this setting, communication is usually the most expensive operation, since it directly connects to network bandwidth usage and energy consumption.
\end{remark}

\subsection{Our Results}
We investigate lower bounds (impossibility results) and upper bounds (protocols) of the exact computation of the following basic statistical and graph problems in the message-passing model.
\begin{enumerate}
\item Statistical problems: computing the number of distinct elements, known as $F_0$ in the database literature; and finding the element with the maximum frequency, known as the $\ell_\infty$ or iceberg query problem. 
We note that the lower bound for $\ell_\infty$ also applies to the heavy-hitter problem of finding all elements whose frequencies exceed a certain threshold, as well as many other statistical problems for which we have to compute the elements with the maximum frequency exactly.

\item Graph problems: computing the degree of a vertex; testing cycle-freeness; testing connectivity; computing the number of connected components (\#CC); testing bipartiteness; and testing triangles-freeness. 
%and clique-freeness for cliques of constant size.
\end{enumerate}
For each graph problem, we study its lower bound and upper bound in two cases: with edge duplication among the different sites and without edge duplication. Our results are summarized in Table~\ref{tab:results}. Note that all lower bounds are matched by upper bounds up to some logarithmic factors. For convenience, we use $\tilde{\Omega}(f)$ and $\tilde{O}(f)$ to denote functions of forms $f / \log^{O(1)}(f)$ and $f \cdot \log^{O(1)}(f)$, respectively. That is, we hide logarithmic factors.

\begin{table*}[t!]
\centering
\scalebox{1}{
\begin{tabular}{|c|c|c|c|c|}
\hline
& \multicolumn{2}{|c|}{With duplication} & \multicolumn{2}{|c|}{Without duplication}\\
\hline
Problem & LB  &  UB & LB & UB \\
\hline
$F_0$   & $\tilde{\Omega}(k F_0)$ & $\tilde{O}(k (F_0 + \log n))$ & -- & -- \\ 
$\ell_\infty$  & $\tilde{\Omega}(\min\{k, \ell_\infty\} n)$ & $\tilde{O}(\min\{k, \ell_\infty\} n)$ & -- & -- \\
degree & $\tilde{\Omega}(k d_v)$ & $O(k d_v \log n)$ & $\tilde{\Omega}(k)$ & $O(k \log n)$  \\
cycle-freeness & $\tilde{\Omega}(kn)$ & $\tilde{O}(kn)$ & ${\Omega}(\min\{n, m\})$ & $\tilde{O}(\min\{n,m\})$  \\
connectivity & $\tilde{\Omega}(kn)$ & $\tilde{O}(kn)$ & $\tilde{\Omega}(kr)$ & $\tilde{O}(kr)$  \\
\#CC & $\tilde{\Omega}(kn)$ & $\tilde{O}(kn)$ & $\tilde{\Omega}(kr)$ & $\tilde{O}(kr)$  \\
bipartiteness & $\tilde{\Omega}(kn)$ & $\tilde{O}(kn)$ & $\tilde{\Omega}(kr)$ & $\tilde{O}(kr)$  \\
triangle-freeness & $\tilde{\Omega}(km)$ & $\tilde{O}(km)$ & ${\Omega}(m)$ & $\tilde{O}(m)$  \\
%clique-freeness & $\Omega(km)$ & $O(km \log n)$ & $\Omega(m)$ & $O(m \log n)$  \\
\hline
\end{tabular}
}
\caption{
All results are in terms of number of bits of communication. Our lower bounds hold for randomized protocols which succeed with at least a constant probability of $2/3$, while all of our upper bounds are deterministic protocols (which always succeed). $k$ refers to the number of sites, with a typical value ranging from $100$ to $10000$ in practice. For $F_0$ and $\ell_\infty$, $n$ denotes the size of the element universe. For graph problems, $n$ denotes the number of vertices and $m$ denotes the number of edges. $d_v$ is the degree of the queried vertex $v$. 
We make the mild assumption that $\Omega(\log n) \le k \le \min\{n, m\}$. Let $r = \min\{n, m/k\}$. Except for the upper bound for cycle-freeness in the ``without duplication" case, for which $m \ge n$ implies that a cycle necessarily exists (and therefore makes the problem statement vacuous), we assume that $m \ge n$ in order to avoid a messy and uninteresting case-by-case analysis.}
\label{tab:results}
\end{table*}

We prove most of our lower bound results via reductions from a meta-problem that we call \thresh\ . Its definition is given in Section~\ref{sec:thresh}. 

In Section \ref{sec:diameter} we make a conjecture on the lower bound for the diameter problem, i.e., the problem of computing the distance of the farthest pair of vertices in a graph. This problem is one of the few problems that we cannot completely characterize by the technique proposed in this paper. We further show that by allowing an error as small as an additive-$2$, we can reduce the communication cost of computing the diameter by roughly a $\sqrt{n}$ factor, compared with the naive algorithm for exact computation. This further supports our claim that even a very slight approximation can result in a dramatic savings in communication. 

\subsection{Related Work}
\label{sec:related-work}
For statistical problems, a number of approximation algorithms have been proposed recently in the {\em distributed streaming} model, which can be thought of as a dynamic version of the one-shot distributed computation model considered in this paper: the $k$ local inputs arrive in the streaming fashion and one of the sites has to continuously monitor a function defined on the union of the $k$ local inputs. All protocols in the distributed streaming model are also valid protocols in our one-shot computational model, while our impossibility results in our one-shot computational model also apply to all protocols in the distributed streaming model. Example functions studied in the distributed streaming model include $F_0$ \cite{CMY11}, $F_2$ (size of self join) \cite{CMY11,WZ12}, quantile and heavy-hitters \cite{HYZ12}, and the empirical entropy~\cite{ABC09}. All of these problems have much lower communication cost if one allows an approximation of the output number $x$ in a range $[(1-\eps)x, (1+\eps)x]$, as mentioned above (the definition as to what $\eps$ is for the various problems differs). These works show that if an approximation is allowed, then all these problems can be solved using only $\tilde{O}(k/\eps^{O(1)})$ bits of communication. A suite of (almost) matching lower bounds for approximate computations was developed in \cite{WZ12}. For exact $F_0$ computation, the best previous communication cost lower bound was $\Omega(F_0+k)$ bits. In this paper we improve the communication cost lower bound to $\tilde{\Omega}(k F_0)$, which is optimal up to a small logarithmic factor. 

%We remark that computing $F_0$ exactly also allows one to determine whether an item is missing from the union of the $k$ data sets, a problem whose complexity has been studied in the data streaming literature, see, e.g., \cite{t07}.

For graph problems, Ahn, Guha and McGregor~\cite{AGM12a,AGM12b} developed an elegant technique for {\it sketching} graphs, and showed its applicability to many graph problems including connectivity, bipartiteness, and minimum spanning tree. Each sketching step in these algorithms can be implemented in the message-passing model as follows: each site computes a sketch of its local graph and sends its sketch to $P_1$. The site $P_1$ then combines these $k$ sketches into a sketch of the global graph. The final answer can be obtained based on the global sketch that $P_1$ computes. Most sketches in \cite{AGM12a,AGM12b} are of size $\tilde{O}(n^{1+\gamma})$ bits (for a small constant $\gamma \ge 0$), and the number of sketching steps varies from $1$ to a constant. Thus direct implementations of these algorithms in the message-passing model have communication $\tilde{O}(k \cdot n^{1+\gamma})$ bits. On the lower bound side, it seems not much is known. Phillips et al. \cite{PVZ12} proved an $\Omega(kn/\log^2 k)$ bits lower bound for connectivity. Their lower bound proof relies on a well-crafted graph distribution. In this paper we improve their lower bound by a factor of $\log k$. Another difference is that their proof requires the input to have edge duplications, while our lower bound holds even if there are no edge duplications, showing that the problem is hard even if each edge occurs on a single site. Very recently in an unpublished manuscript, Huang et. al. \cite{HRVZ13} showed that $\Omega(kn)$ bits of communication is necessary in order to even compute a constant factor approximation to the size of the maximum matching of a graph. Their result, however, requires that the entire matching has to be reported, and it is unknown if a similar lower bound applies if one is only interested in estimating the matching size. 

Besides statistical and graph problems, Koutris and Suciu \cite{KS11} studied evaluating conjunctive queries in their massively parallel model. Their lower bounds are restricted to one round of communication, and the message format has to be tuple-based, etc. Some of these assumptions are removed in a recent work by Beame et al.~\cite{BKS13}.
We stress that in our message-passing model there is no such restriction on the number of rounds and the message format; our lower bounds apply to {\it arbitrary} communication protocols. Recently, Daum\'e III et al. \cite{DPSV12a,DPSV12b} and Balcan et al. \cite{BBFM12} studied several problems in the setting of distributed learning, in the message-passing model. 

%Finally, quite a few graph problems have been studied recently in the MapReduce model, including computing the connected components \cite{LMSV11,RMCS12}, counting the number of triangles \cite{SV11,ASSU12}, finding a minimum spanning tree \cite{LMSV11}, computing a matching and an edge cover \cite{LMSV11}, finding a minimum-cut \cite{LMSV11}, and finding densest subgraphs~\cite{BKV12}. The primary goal of all these works, however, is to minimize the round complexity, given various constraints on the amount of messages that each site can send/receive at each round. 

\subsection{Conventions}
Let $[n] = \{1, \ldots, n\}$. All logarithms are base-$2$.  All communication complexities are in terms of bits. We typically use capital letters $X, Y, \ldots$ for sets or random variables, and lower case letters $x, y, \ldots$ for specific values of the random variables $X, Y, \ldots$. We write $X \sim \mu$ to denote a random variable chosen from distribution $\mu$. For convenience we often identify a set $X \subseteq [n]$ with its characteristic vector when there is no confusion, i.e., the bit vector which is $1$ in the $i$-th bit if and only if element $i$ occurs in the set $X$. 

All our upper bound protocols are either deterministic or only using private randomness. 

We make a mild assumption that $\Omega(\log n) \le k \le \min\{n, m\}$, where for $F_0$ and $\ell_\infty$, $n$ denotes the size of the element universe; and for graph problems, $n$ denotes the number of vertices and $m$ denotes the number of edges.

%Let $H(X)$ denote the usual Shannon entropy of a random variable $X$, that is, $H(X) = \sum_x \Pr[X = x] \log(1/\Pr[X = x])$. And let $H(X\ |\ Y)$ denote the conditional entropy, that is, $H(X\ |\ Y) = \sum_y \Pr[Y = y] H(X\ |\ Y = y)$. Let $H_b(p)$ denote the binary entropy function when $p \in [0, 1]$, that is, $H_b(p) = p \log (1/p) + (1-p) \log(1/(1-p))$ with the usual convention that $H_b(0) = H_b(1) = 0$. 

\subsection{Roadmap}
In Section~\ref{sec:F0}, we give a case study on the number of distinct elements ($F_0$) problem. In Section~\ref{sec:prelim}, we include background on communication complexity which is needed for understanding the rest of the paper.
In Section~\ref{sec:thresh}, we introduce the meta-problem \thresh\ and study its communication complexity. In Section~\ref{sec:stat} and Section~\ref{sec:graph}, we show how to prove lower bounds for a set of statistical and graph problems by performing reductions from \thresh. We conclude the paper in Section~\ref{sec:discussion}. 
%All missing proofs can be found in the Appendix.

\section{The Number of Distinct Elements: A Case Study}
\label{sec:F0}
In this section we give a brief case study on the number of distinct elements ($F_0$) problem, with the purpose of justifying the statement that approximation is often needed in order to obtain communication-efficient protocols in the distributed setting.

The $F_0$ problem requires computing the number of distinct elements of a data set. It has numerous applications in network traffic monitoring~\cite{EVF06}, data mining in graph databases~\cite{PGF02}, data integration~\cite{BHMPRS05}, etc., and has been extensively studied in the last three decades, mainly in the data stream model. It began with the work of Flajolet and Martin \cite{flajolet85:_probab} and culminated in an optimal algorithm by Kane et al. \cite{Kane10}. In the streaming setting, we see a stream of elements coming one at a time and the goal is to compute the number of distinct elements in the stream using as little memory as possible. In \cite{flajolet2008}, Flajolet et al. reported that their HyperLogLog algorithm can estimate cardinalities beyond $10^9$ using a memory of only $1.5$KB, and achieve a relative accuracy of $2\%$, compared with the $10^9$ bytes of memory required if we want to compute $F_0$ exactly.

Similar situations happen in the distributed communication setting, where we have $k$ sites, each holding a set of elements from the universe $[n]$, and the sites want to compute the number of distinct elements of the union of their $k$ data sets. In \cite{CMY11}, a $(1+\eps)$-approximation algorithm (protocol) with $O(k (\log n + 1/\eps^2 \log 1/\eps))$ bits of communication was given in the distributed streaming model, which is also a protocol in the message-passing model. In a typical setting, we could have $\eps = 0.01$, $n = 10^9$ and $k = 1000$, in which case the communication cost is about $6.6 \times 10^7$ bits~\footnote{In the comparison we neglect the constants hidden in the big-$O$ and big-$\Omega$ notation which should be small.}. On the other hand, our result shows that if exact computation is required, then the communication cost among the $k$ sites needs to be at least be $\Omega(k F_0 / \log k)$ (See Corollary~\ref{cor:F0}), which is already $10^{9}$ bits even when $F_0 = n/100$.

\section{Preliminaries}
\label{sec:prelim}

%\qinsays{I have written some clich\'e for DB people}

In this section we introduce some background on communication complexity. We refer the reader to the book by Kushilevitz and Nisan \cite{KN97} for a more complete treatment.

In the basic two-party communication complexity model, we have two parties (also called sites or players), which we denote by Alice and Bob. Alice has an input $x$ and Bob has an input $y$, and they want to jointly compute a function $f(x,y)$ by communicating with each other according to a protocol $\Pi$. Let $\Pi(x, y)$ be the transcript of the protocol, that is, the concatenation of the sequence of messages exchanged by Alice and Bob, given the inputs $x$ and $y$. In this paper when there is no confusion, we abuse notation by using $\Pi$ for both a protocol and its transcript, and we further abbreviate the transcript $\Pi(x,y)$ by $\Pi$.

The {\em deterministic communication complexity} of a deterministic protocol is defined to be \\
$\max\{\abs{\Pi(x,y)}\ |\ \textrm{all possible inputs $(x,y)$}\}$, where $\abs{\Pi(x,y)}$ is the number of bits in the transcript of the protocol $\Pi$ on inputs $x$ and $y$. The {\em randomized communication complexity} of a randomized protocol $\Pi$ is the maximum number of bits in the transcript of the protocol over all possible inputs $x,y$, together with all possible random tapes of the players. We say a randomized protocol $\Pi$ computes a function $f$ correctly with error probability $\delta$ if for all input pairs $(x,y)$, it holds that $\Pr[\Pi(x,y) \neq f(x,y)] \le \delta$, where the probability is taken only over the random tapes of the players. The randomized {\em $\delta$-error communication complexity} of a function $f$, denoted by $R^{\delta}(f)$, is the minimum communication complexity of a protocol that computes $f$ with error probability at most $\delta$.

Let $\mu$ be a distribution over the input domain, and let $(X, Y) \sim \mu$. For a deterministic protocol $\Pi$, we say that $\Pi$ computes $f$ with error probability $\delta$ on $\mu$ if $\Pr[\Pi(X,Y) \neq f(X,Y)]  \le \delta$, where the probability is over the choices of $(X,Y) \sim \mu$. The {\em $\delta$-error $\mu$-distributional communication complexity} of $f$, denoted by $D^{\delta}_{\mu}(f)$, is the minimum worst-case communication complexity of a deterministic protocol that gives the correct answer for $f$ on at least $(1 - \delta)$ fraction of all inputs (weighted by $\mu$). We denote ${\sf ED}^{\delta}_{\mu}(f)$ to be the {\em $\delta$-error $\mu$-distributional  expected  communication complexity}, which is define to be the minimum expected cost (rather than the worst-case cost) of a deterministic protocol that gives the correct answer for $f$ on at least $(1 - \delta)$ fraction of all inputs (weighted by $\mu$), where the expectation is taken over distribution $\mu$.

We can generalize the two-party communication complexity to the multi-party setting, which is the message-passing model considered in this paper. Here we have $k$ players (also called sites) $P_1, \ldots, P_k$ with $P_j$ having the input $x_j$, and the players want to compute a function $f(x_1, \ldots, x_k)$ of their joint inputs by exchanging messages with each other. 
The transcript of a protocol always specifies which player speaks next. 
In this paper the communication is point-to-point, that is, if $P_i$ talks to $P_{j}$, the other players do not see the messages sent from $P_i$ to $P_j$. At the end of the communication, only one player needs to output the answer. 

The following lemma shows that randomized communication complexity is lower bounded by distributional communication complexity under any distribution $\mu$. We include a proof in Appendix~\ref{sec:proof-Yao}, since the original proof is for the two-party communication setting.
\begin{lemma}[Yao's Lemma \cite{yao77}]
\label{lem:Yao}
For any function $f$ and any $\delta > 0$,
$R^{\delta}(f) \geq \max_{\mu} D^{\delta}_\mu(f).$
\end{lemma}

Therefore, one way to prove a lower bound on the randomized communication complexity of $f$ is to first pick a (hard) input distribution $\mu$ for $f$, and then study its distributional communication complexity under $\mu$.

Note that given a $1/3$-error randomized protocol for a problem $f$ whose output is $0$ or $1$, we can always run the protocol $C\log(1/\delta)$ times using independent randomness each time, and then output the majority of the outcomes. By a standard Chernoff bound (see below), the output will be correct with error probability at most $e^{-\kappa C \log(1/\delta)}$ for an absolute constant $\kappa$, which is at most $\delta$ if we choose $C$ to be a sufficiently large constant. Therefore $R^{1/3}(f) = \Omega(R^{\delta}(f)/\log(1/\delta)) = \Omega(\max_\mu D^{\delta}_\mu(f)/\log(1/\delta))$ for any $\delta \in (0, 1/3]$. Consequently, to prove a lower bound on $R^{1/3}(f)$ we only need to prove a lower bound on the distributional communication complexity of $f$ with an error probability $\delta \le 1/3$.
\medskip

\noindent{\bf Chernoff bound.} Let $X_1, \ldots, X_n$ be independent Bernoulli random variables such that $\Pr[X_i = 1] = p_i$. Let $X = \sum_{i \in [n]} X_i$. Let $\mu = \E[X]$. It holds that $\Pr[X \ge (1+\delta)\mu] \le e^{-\delta^2\mu/3}$ and $\Pr[X \le (1-\delta)\mu] \le e^{-\delta^2\mu/2}$ for any $\delta \in (0,1)$.

%We also need the following variant of Hoeffding bound.
%\begin{lemma}[A variant of the Hoeffding bound]
% \label{lem:Hoeffding}
%Let $X_1, \ldots, X_n$ be $n$ independent random variables
%such that $X_i \in [0,1]$. Let $\mu = \E[\sum_{i \in [n]} X_i]$. Then for any $a > 0$, we have
%$\Pr\left[\sum_{i\in[n]} X_i > a\right] \le e^{-(a-2\mu)}.$
%\end{lemma}

%\paragraph{An anti-concentration lemma.}
%In the analysis we need the following anti-concentration result which is an easy consequence of Feller~\cite{feller:43} (cf. \cite{Matousek:08}).
%\begin{lemma}(\cite{Matousek:08})
%\label{lem:feller}
%Let $X$ be a sum of independent random variables, each attaining values in $[0,1]$, and let $\sigma = \sqrt{\var[X]} \ge 200$. Then for all $t \in [0, \sigma^2/100]$, we have
%$$\Pr[X \ge \E[X] + t] \ge \tilde{c} \cdot e^{-t^2/(3\sigma^2)}$$ for a universal constant $\tilde{c} > 0$.
%\end{lemma}

\section{A Meta-Problem}
\label{sec:thresh}

\vspace{-0.4cm}
In this section we discuss a meta-problem \thresh\ and we derive a communication lower bound for it. 
This meta-problem will be used to derive lower bounds for statistical and graph problems in our applications. 
\smallskip

In the \thresh\ problem, site $P_i\ (i \in [k])$ holds an $r$-bit vector $x_i = \{x_{i,1}, \ldots, x_{i,r}\}$, and the $k$ sites want to compute 
{\small
\begin{eqnarray*}
\begin{array}{l}
\text{\thresh}(x_1, \ldots, x_k)  =  \left\{
  \begin{array}{rl}
   0, & \text{if } \sum_{j \in [r]} (\vee_{i \in [k]} x_{i, j}) \leq  \theta,\\
   1, & \text{if } \sum_{j \in [r]} (\vee_{i \in [k]} x_{i, j}) \ge  \theta + 1.
  \end{array}
  \right.
\end{array}
\end{eqnarray*}
}That is, if we think of the input as a $k \times r$ matrix with $x_1, \ldots, x_k$ as the rows, then in the \thresh\ problem we want to find out whether the number of columns that contain a $1$ is more than $\theta$ for a threshold parameter $\theta$.

We will show a lower bound for \thresh\ using the symmetrization technique introduced in \cite{PVZ12}. First, it will be convenient for us to study the problem in the {\em coordinator} model.

\paragraph{The Coordinator Model.}
In this model we have an additional site called the coordinator~\footnote{We can also choose, for example, $P_1$ to be the coordinator and avoid the need for an additional site, though having an additional site makes the notation cleaner.}, which has no input (formally, his input is the empty set). We require that the $k$ sites can only talk to the coordinator. The message-passing model can be simulated by the coordinator model since every time a site $P_i$ wants to talk to $P_j$, it can first send the message to the coordinator, and then the coordinator can forward the message to $P_j$. Such a re-routing only increases the communication complexity by a factor of $2$ and thus will not affect the asymptotic communication complexity. 
%Sometimes we also assign the coordinator an ``input". This is fine if such an input can be computed via communication with the $k$ sites. This extra communication cost can be neglected if it is asymptotically smaller than the communication complexity of the problem we consider.

Let $f: \mathcal{X} \times \mathcal{Y} \to \{0,1\}$ be an arbitrary function. Let $\mu$ be a probability distribution over $\mathcal{X} \times \mathcal{Y}$. Let $f_{\OR}^k: \mathcal{X}^k \times \mathcal{Y} \to \{0,1\}$ be the problem of computing $f(x_1, y) \vee f(x_2, y) \vee \ldots \vee f(x_k, y)$ in the coordinator model, where $P_i$ has input $x_i \in \mathcal{X}$ for each $i \in [k]$, and the coordinator has $y \in \mathcal{Y}$. Given the distribution $\mu$ on $\mathcal{X} \times \mathcal{Y}$, we construct a corresponding distribution $\nu$ on $\mathcal{X}^k \times \mathcal{Y}$: We first pick $(X_1, Y) \sim \mu$, and then pick $X_2, \ldots, X_k$ from the conditional distribution $\mu\ |\ Y$. 

The following theorem was originally proposed in \cite{PVZ12}. Here we improve it by a $\log k$ factor by a slightly modified analysis, which we include here for completeness.

\begin{theorem}
\label{thm:or-f}
For any function $f: \mathcal{X} \times \mathcal{Y} \to \{0,1\}$ and any distribution $\mu$ on $\mathcal{X} \times \mathcal{Y}$ for which $\mu(f^{-1}(1)) \le 1/k^2$, we have
$D_\nu^{1/k^3} (f_{\OR}^k) = \Omega(k \cdot {\sf ED}_\mu^{1/(100k^2)}(f))$.
\end{theorem}

\begin{proof}
Suppose Alice has $X$ and Bob has $Y$ with $(X, Y) \sim \mu$, and they want to compute $f(X, Y)$. They can use a protocol $\cal P$ for $f^k_{\OR}$ to compute $f(X, Y)$ as follows. The first step is an input reduction. Alice and Bob first pick a random $I \in [k]$ using shared randomness, which will later be fixed by the protocol to make it deterministic. Alice simulates $P_I$ by assigning it an input $X_I = X$. Bob simulates the coordinator and the remaining $k - 1$ players. He first assigns $Y$ to the coordinator, and then samples $X_1, \ldots, X_{I-1}, X_{I+1}, \ldots, X_k$ independently according to the conditional distribution $\mu\ |\ Y$, and assigns $X_i$ to $P_i$ for each $i \in [k] \backslash I$. Now $\{X_1, \ldots, X_k, Y\} \sim \nu$. Since $\mu(f^{-1}(1)) \le 1/k^2$, with probability $(1 - 1/k^2)^{k-1} \ge 1 - 1/k$, we have $f(X_i, Y) = 0$ for all $i \in [k] \backslash I$. Consequently, 
\begin{equation}
\label{eq:reduction}
f_{\OR}^k(X_1, \ldots, X_k, Y) = f(X, Y).
\end{equation} 
We say such an input reduction is {\em good}.

Alice and Bob construct a protocol $\cal P'$ for $f$ by independently repeating the input reduction three times, and running $\cal P$ on each input reduction. The probability that at least one of the three input reductions is good is at least $1 - 1/k^3$, and Bob can learn which reduction is good without any communication. This is because in the simulation he locally generates all $X_i\ (i \in [k] \backslash I)$ together with $Y$. On the other hand, by a union bound, the probability that $\cal P$ is correct for all three input reductions is at least $1-3/k^3$. Note that if we can compute $f_{\OR}^k(X_1, \ldots, X_k, Y)$ correctly for a good input reduction, then by (\ref{eq:reduction}), $\cal P$ can also be used to correctly compute $f(X, Y)$. Therefore $\cal P$ can be used to compute $f(X, Y)$ with probability at least $1 - 3/k^3 - 1/k^3 \ge 1 - 1/(100k^2)$.
 
Since in each input reduction, $X_1, \ldots, X_k$ are independent and identically distributed, and since $I \in [k]$ is chosen randomly in the two input reductions, we have that in expectation over the choice of $I$, the communication between $P_I$ and the coordinator is at most a $2/k$ fraction of the expected total communication of $\cal P$ given inputs drawn from $\nu$. By linearity of expectation, if the expected communication cost of $\cal P$ for solving $f_{\OR}^k$ under input distribution $\nu$ with error probability at most $1/k^3$ is $C$, then the expected communication cost of $\cal P'$ for solving $f$ under input distribution $\mu$ with error probability at most $1/(100k^2)$ is $O(C/k)$. Finally, by averaging there exists a fixed choice of $I \in [k]$, so that $\cal P'$ is deterministic and for which the expected communication cost of $\cal P'$ for solving $f$ under input distribution $\mu$ with error probability at most $1/(100k^2)$ is $O(C/k)$. Therefore we have $D_\nu^{1/k^3} (f_{\OR}^k) = \Omega(k \cdot {\sf ED}_\mu^{1/(100k^2)}(f))$.
\end{proof}

\subsection{The \DISJ\ Problem}
\label{sec:DISJ}
Now we choose a concrete function $f$ to be the set-disjointness problem. In this problem we have two parties: Alice has $x \subseteq [r]$ while Bob has $y \subseteq [r]$, and the parties want to compute
\begin{eqnarray*}
\text{\DISJ}(x, y)  =  \left\{
  \begin{array}{rl}
   1, & \text{if } x \cap y \neq \emptyset, \\
   0,  & \text{otherwise}.
  \end{array}
  \right.
\end{eqnarray*}
Set-disjointness is a classical problem used in proving communication lower bounds. 
%For example, in a database context it was recently used in proving streaming lower bounds for finding dense subgraphs \cite{BKV12}.
We define an input distribution $\tau_\beta$ for \DISJ\ as follows. Let $\ell = (r+1)/4$. With probability $\beta$, $x$ and $y$ are random subsets of $[r]$ such that $\abs{x} = \abs{y} = \ell$ and $\abs{x \cap y} = 1$, while with probability $1 - \beta$, $x$ and $y$ are random subsets of $[r]$ such that $\abs{x} = \abs{y} = \ell$ and $x \cap y = \emptyset$.  Razborov~\cite{Raz90} proved that for $\beta = 1/4$, $D^{(1/4)/100}_{\tau_{1/4}}(\mbox{\DISJ}) = \Omega(r)$, and one can extend his arguments to any $\beta \in (0, 1/4]$, and to the expected distributional communication complexity where the expectation is take over the input distribution.
\begin{theorem}[\cite{PVZ12}, Lemma 2.2]
\label{thm:DISJ}
For any $\beta \in (0, 1/4]$,
it holds that ${\sf ED}^{\beta/100}_{\tau_\beta}(\mbox{\DISJ}) = \Omega(r)$, where the expectation is taken over the input distribution.
\end{theorem}

\subsection{The \orDISJ\ Problem}
If we choose $f$ to be \DISJ\ and let $\mu = \tau_\beta$ with $\beta = 1/k^2$, then we call $f_{\OR}^k$ in the coordinator model the \orDISJ\ Problem. By Theorem~\ref{thm:or-f} and Theorem~\ref{thm:DISJ}. We have
\begin{theorem}
\label{thm:orDISJ}
$D_\nu^{1/k^3} (\text{\orDISJ}) = \Omega(kr)$.
\end{theorem}

\subsubsection{The Complexity of \thresh}
We prove our lower bound for the setting of the parameter $\theta = (3r-1)/4$. We define the following input distribution $\zeta$ for \threshR: We choose $\{X_1, \ldots, X_k, Y\} \sim \nu$ where $\nu$ is the input distribution for \orDISJ, and then simply use $\{X_1, \ldots, X_k\}$ as the input for \thresh.

\begin{lemma}
\label{lem:noise}
Under the distribution $\zeta$, assuming $k \ge c_k \log r$ for a large enough constant $c_k$, we have that $\bigvee_{i \in [k]} X_{i,j} = 1$ for all $j \in [r] \backslash Y$ with probability $1 - 1/k^{10}$. 
\end{lemma}
\begin{proof}
For each $j \in [r] \backslash Y$, we have $\bigvee_{i \in [k]} X_{i,j} = 1$ with probability at least $1 - (1-1/4)^k$. This is because $\Pr[X_{i,j} = 1] \ge 1/4$ for each $j \in [r] \backslash Y$, by our choices of $X_i$. By a union bound, with probability at least 
\begin{eqnarray*}
&&\left(1 - (3/4)^k \cdot \abs{[r] \backslash Y} \right) \\
&=& \left(1 - (3/4)^k \cdot (3r-1)/4 \right) \\
&\ge&  1 - 1/k^{10}
\end{eqnarray*} (by our assumption $c_k\log r \le k \le r$ for a large enough constant $c_k$), we have $\bigvee_{i \in [k]} X_{i,j} = 1$ for all $j \in [r] \backslash Y$.
\end{proof}

\begin{theorem}
\label{thm:thresh}
$D^{1/k^4}_{\zeta}(\mbox{\threshR}) = \Omega(kr)$, assuming $c_k \log r \le k \le r$ for a large enough constant $c_k$.
\end{theorem}
\begin{proof}
By Lemma~\ref{lem:noise}, it is easy to see that any protocol $\mathcal{P}$ that computes \threshR\ on input distribution $\zeta$ correctly with error probability $1/k^4$ can be used to compute \orDISJ\ on distribution $\nu$ correctly with error probability $1/k^4 +  1/k^{10} < 1/k^3$, since if $(X_1, \ldots, X_k, Y) \sim \nu$, then with probability $1-1/k^{10}$, we have
$$\textstyle \text{\orDISJ}(X_1, \ldots, X_k, Y) = \left( \exists j \in Y, \text{\DISJ}(X_j, Y) = 1 \right) = \text{\threshR}(X_1, \ldots, X_k).$$
The theorem follows from Theorem~\ref{thm:orDISJ}.
\end{proof}

%\section{Applications}
%In this section we show how our lower bounds for \thresh\ implies communication lower bounds for a set of basic statistical and graph problems in the message-passing model.

\section{Statistical Problems}
\label{sec:stat}
For technical convenience in the reductions, we make the mild assumption that $c_k \log n \le k \le n$ where $c_k$ is some large enough constant. For convenience, we will repeatedly ignore an additive $ O(1/k^{10})$ error probability introduced in the reductions, since these will not affect the correctness of the reductions, and can be added to the overall error probability by a union bound. 

\subsection{$F_0$ (\#distinct-elements)}
Recall that in the $F_0$ problem, each site $P_i$ has a set $S_i \subseteq [n]$, and the $k$ sites want to compute the number of distinct elements in $\bigcup_{i \in [k]} S_i$. 

For the lower bound, we reduce from \threshP. Given an input $\{X_1, \ldots, X_k\} \sim \zeta$ for \threshP, each site sets $S_i = X_i$.  Let $\sigma_F$ be the input distribution of $F_0$ after this reduction.

By Lemma~\ref{lem:noise} we know that under distribution $\zeta$, with probability $1 - 1/k^{10}$, for all $j \in [n] \backslash Y$ (recall that $Y$ is the random subset of $[n]$ of size $(n+1)/4$ we used to construct $X_1, \ldots, X_k$ in distribution $\zeta$), $\bigvee_{i \in [k]} X_{i,j} = 1$. Conditioned on this event, we have 
\begin{eqnarray*}
&&\text{\threshP}(X_1, \ldots, X_k) = 1 \\ &\iff &F_0(\cup_{i \in [k]} S_i) > (3n-1)/4.
\end{eqnarray*}
Therefore, by Theorem~\ref{thm:thresh} we have that $D^{1/k^4}_{\sigma_F}(F_0) = \Omega(kn)$. Note that in this reduction, we have to choose $n = \Theta(F_0)$. Therefore, it makes more sense to write the lower bound as $D^{1/k^4}_{\sigma_F}(F_0) = \Omega(k F_0)$.

%The following corollary is similar to Corollary~\ref{cor:thresh}.
The following corollary follows from Yao's Lemma (Lemma~\ref{lem:Yao}) and the discussion following it. 
\begin{corollary}
\label{cor:F0}
$R^{1/3}(F_0) = \Omega(k F_0/\log k)$.
%, assuming $c_k \log F_0 \le k \le F_0$ for a large enough constant $c_k$.
\end{corollary}
%Similar corollaries hold for all other problems in Section~\ref{sec:stat} and Section~\ref{sec:graph}, using Yao's Lemma, and we  will not explicitly state such corollaries. 

An almost matching upper bound of $O(k (F_0 \log F_0 + \log n))$ can be obtained as follows: the $k$ sites first compute a $2$-approximation $F_0'$ to $F_0$ using the protocol in \cite{CMY11} (see Section~\ref{sec:F0}), which costs $O(k \log n)$ bits. Next, they hash every element to a universe of size $(F_0')^3$, so that there are no collisions among hashed elements with probability at least $1 - 1/F_0$, by a union bound. Finally, all sites send their distinct elements (after hashing) to $P_1$ and then $P_1$ computes the number of distinct elements over the union of the $k$ sets locally. This step costs $O(k F_0\log F_0)$ bits of communication.

\subsection{$\ell_\infty$ (MAX)}
In the $\ell_\infty$ problem, each site $P_i$ has a set $S_i \subseteq [n]$, and the $k$ sites want to find an element in $\bigcup_{i \in [k]} S_i$ with the maximum frequency. 

For the lower bound, we again reduce from \threshP. Recall that in our hard input distribution for \threshP, there is one special column that contains zero or a single $1$. The high level idea is that we try to make this column to have the maximum number of $1$'s if originally it contains a single $1$, by flipping bits over a random set of rows. Concretely, given an input $\{X_1, \ldots, X_k\} \sim \zeta$ for \threshP, the $k$ sites create an input $\{S_1, \ldots, S_k\}$ as follows: first, $P_1$ chooses a set $R \subseteq [k]$ by independently including each $i \in [k]$ with probability $7/8$, and informs all sites $P_i\ (i \in R)$ by sending each of them a bit. This step costs $O(k)$ bits of communication. Next, for each $i \in R$, $P_i$ flips $X_{i,j}$ for each $j \in [n]$. Finally, each $P_i$ includes $j \in S_i$ if $X_{i,j} = 1$ after the flip and $j \not\in S_i$ if $X_{i,j} = 0$. Let $\sigma_L$ be the input distribution of $\ell_\infty$ after this reduction. 

They repeat this input reduction independently $T$ times where $T = c_T \log k$ for a large constant $c_T$, and at each time they run $\ell_\infty(\cup_{i \in [k]} S_i)$. Let $R_1, \ldots, R_T$ be the random set $R$ sampled by $P_1$ in the $T$ runs, and let $O_1, \ldots, O_T$ be the outputs of the $T$ runs. They return \threshP$(X_1, \ldots, X_k) = 1$ if there exists a $t \in [T]$ such that $O_t \ge \abs{R_t} + 1$ and $0$ otherwise. 

We focus on a particular input reduction. We view an input for \threshP\ as a $k \times n$ matrix. The $i$-th row of the matrix is $X_i$. After the bit-flip operations, for each column $j \in [n] \backslash Y$, we have for each $i \in [k]$ that 
\begin{eqnarray*}
&&\Pr[X_{i,j} = 1] \\
&\le& 7/8 \cdot \left(1 - \frac{(n+1)/4 - 1}{(3n-1)/4}\right) + 1/8 \cdot \frac{(n+1)/4}{(3n-1)/4} \\
&<& 3/4.
\end{eqnarray*}
By a Chernoff bound, for each $j \in [n] \backslash Y$, $\sum_{i \in [k]} X_{i,j} < 13k/16$ with probability $1 - e^{-\Omega(k)}$. Therefore with probability at least $(1 - e^{-\Omega(k)} \cdot n) \ge (1 - 1/k^{10})$ (assuming that $c_k \log n \le k \le n$ for a large enough constant $c_k$), $\sum_{i \in [k]} X_{i,j} < 13k/16$ holds for all $j \in [n] \backslash Y$. 

Now we consider columns in $Y$. We can show again by Chernoff bound that $\abs{R} > 13k/16$ with probability $(1 - 1/k^{10})$ for all columns in $Y$, since each $i \in [k]$ is included into $R$ with probability $7/8$, and before the flips, the probability that $X_{i,j} = 1$ for an $i$ when $j \in Y$ is negligible. Therefore with probability $(1 - 1/k^{10})$, the column with the maximum number of $1$s is in the set $Y$, which we condition on in the rest of the analysis.

In the case when \threshP$(X_1, \ldots, X_k) = 1$, then with probability at least $1/8$, there exists a column $j \in Y$ and a row $i \in [k] \backslash R$ for which $X_{i,j} = 1$. If this happens, then for this $j$ we have $\sum_{i \in [k]} X_{i,j} \ge \abs{R} + 1$, or equivalently, $\ell_\infty(\cup_{i \in [k]} S_i) \ge \abs{R} + 1$. Otherwise, if \threshP$(X_1, \ldots, X_k) = 0$, then $\sum_{i \in [k]} X_{i,j} = \abs{R}$ for all $j \in Y$. Therefore, if \threshP$(X_1, \ldots, X_k) = 1$, then the probability that there exists a $t \in [T]$ such that $O_t \ge \abs{R_t}+1$ is at least $1 - (1-1/8)^T > 1 - 1/k^{10}$ (by choosing $c_T$ large enough). Otherwise, if \threshP$(X_1, \ldots, X_k) = 0$, then $O_t = \abs{R_t}$ for all $t \in [T]$.

Since our reduction only uses $T \cdot O(k) = O(k \log k)$ extra bits of communication and introduces an extra error of $O(1/k^{10})$, which will not affect the correctness of the reduction. By Theorem~\ref{thm:thresh}, we have that $D^{1/k^4}_{\sigma_L}(\ell_\infty) = \Omega(kn)$. Note that in the reduction, we have to assume that $\Theta(\ell_\infty) = \Theta(k)$. In other words, if $\ell_\infty \ll k$ then we have to choose $k' = \Theta(\ell_\infty)$ sites out of the $k$ sites to perform the reduction. Therefore it makes sense to write the lower bound as $D^{1/k^4}_{\sigma_L}(\ell_\infty) = \Omega(\min\{\ell_\infty, k\} n)$.

The following corollary follows from Yao's Lemma (Lemma~\ref{lem:Yao}) and the discussion following it. 
\begin{corollary}
\label{cor:l-infty}
$R^{1/3}(\ell_\infty) = \Omega(\min\{\ell_\infty, k\} n/\log k)$.
\end{corollary}

A simple protocol that all sites send their elements-counts to the first site solves $\ell_{\infty}$ with $O(\min\{k, \ell_\infty\} n \log n)$ bits of communication, which is almost optimal in light of our lower bound above. 

\section{Graph Problems}
\label{sec:graph}
In this section we consider graph problems. Let $G = (V, E)$ with $\abs{V} = n$ and $\abs{E} = m$ be an undirected graph. Each site has a subgraph $G_i \subseteq G$, and the $k$ sites want to compute a property of $G$ via a communication protocol. For technical convenience we again assume that $c_k \log n \le k \le \min\{n, m\}$, where $c_k$ is a large enough constant. Except for the upper bound for cycle-freeness in the without edge duplication case, for which $m \ge n$ always causes the graph to not be cycle-free, 
we assume that $m \ge n$ to avoid an uninteresting case-by-case analysis.

Most lower bounds in this section are shown by reductions from \threshR\ for some value $r \le n^2$. For convenience of presentation, during some reductions we may generate graphs with more than $n$ vertices. This will not affect the order of the lower bounds as long as the number of vertices is $O(n)$ and the number of edges is $O(m)$ (if $m$ appears in the lower bound as a parameter). 
%At the end of this section, we post a conjecture on the diameter problem which is left open by this work.

The following procedure will be used several times in our reductions. Thus, we present it separately.
\\

{\bf Reconstructing $Y$ from $X_1, X_2, \ldots, X_k$.} Given an input $\{X_1, \ldots, X_k\} \sim \zeta$ for \threshR\ to the $k$ sites, the first site $P_1$ can construct $Y$ correctly with probability $1 - O(1/k^{10})$, using $O(r \log r)$ bits of communication, assuming that $c_k \log r \le k \le r$ for a large enough constant $c_k$. We view the input as a $k \times r$ matrix with the $k$ sites' inputs as rows. For each column $j \in [r]$, $P_1$ randomly selects $c_Y \log r$ sites for some large enough constant $c_Y$, asks each of them for the $j$-th bit of their input vectors, and then computes the sum of these bits, denoted by $s_j$. Note that if $j \in [r] \backslash Y$, then by a Chernoff bound with probability $1 - e^{-\kappa \cdot c_Y \cdot \log r}$ ($\kappa$ is an absolute constant), we have that $s_j \ge c_Y\log r/2$. Therefore with probability $1 - e^{-\kappa \cdot c_Y \cdot \log r} \cdot r \ge 1 - 1/k^{10}$, for all $j \in [r] \backslash Y$, it holds that $s_j \ge c_Y\log r/2$. On the other hand, again by a Chernoff bound we have that with probability $1 - 1/k^{10}$, for all $j \in Y$, $s_j < c_Y \log r/2$. Therefore $P_1$ can reconstruct $Y$ correctly with probability $1 - O(1/k^{10})$. Since the $O(r \log r)$ extra bits of communication and the $O(1/k^{10})$ additional error probability will not affect the correctness of any of our reductions below, we can simply assume that $P_1$ can always reconstruct $Y$ for free.

\subsection{Degree}
In the degree problem, give a vertex $v \in V$, the $k$ sites want to compute the degree of $v$. 

If edge duplication is not allowed, then the degree problem can be solved in $O(k \log n)$ bits of communication: each site sends the number of edges containing the query vertex to the first site $P_1$ and then $P_1$ adds up these $k$ numbers. A lower bound of $\Omega(k)$ bits also holds since each site has to speak at least once in our communication model.

When we allow edge duplication, the degree problem is essentially the same as that of $F_0$, by the following reduction. Given an input $\{X_1, \ldots, X_k\} \sim \tau_F$ for $F_0$, the $k$ sites construct a graph $G = (V, E)$ where $V = \{v_1, \ldots, v_n\}$ of size $n$. Each site $P_i\ (i \in [k])$ does the following: for each element $j \in X_i$ and $j \neq 1$, it creates an edge $(v_1, v_j)$. Let $G$ be the resulting graph. Then 
$F_0(X_1, \ldots, X_k) = \text{Degree}(v_1) - 1.$
Thus from Corollary~\ref{cor:F0} we get
\begin{corollary}
\label{cor:degree}
$R^{1/3}(\text{degree}) = \Omega(k d_v/\log k)$.
\end{corollary}

An $O(k d_v \log n)$ bit upper bound is the following: each site sends all the neighbouring vertices of $v$ to the first site.

\subsection{Cycle-freeness}
In the cycle-freeness problem, the $k$ sites want to check whether $G$ contains a cycle.

\subsubsection{Without Edge Duplication}
If edge duplication is not allowed, then we have the following simple protocol: $P_2, \ldots, P_k$ send the number of their local edges to $P_1$ and $P_1$ computes the total number of edges in the graph $G$, denoted by $m$. If $m \ge n$  then $P_1$ determines immediately that $G$ contains a cycle, since every graph on $n$ vertices having at least $n$ edges must contain a cycle. Otherwise if $m < n$, then $P_2, \ldots, P_k$ send all their edges to $P_1$, who then does a local check. The communication cost of this protocol never exceeds $O(k \log n +\min\{m, n\}\log n) = O(\min\{m,n\}\log n)$ bits.

Let $h = \min\{m,n\}$. An $\Omega(h)$ bit lower bound holds even when $k = 2$, by a reduction from the \DISJP\ problem: suppose $P_1$ has $X$ and $P_2$ has $Y$, where $(X, Y) \sim \tau_{1/4}$ is the hard input distribution for \DISJP. $P_1$ and $P_2$ construct a graph $G$ on the vertex set $\{s, t, v_1, \ldots, v_h\}$ as follows: for each $i \in X$, $P_1$ creates an edge $(s, v_i)$, and he/she also creates an additional edge $(s, t)$. Similarly, for each $i \in Y$, $P_2$ creates an edge $(v_i, t)$. Let $\sigma_{C_1}$ be the resulting input distribution of $G$. It is easy to see from the reduction that if $X \cap Y \neq  \emptyset$, then there is a cycle in the form of $s \to t \to v_i \to s$ for some $i \in [h]$. Otherwise the graph is a forest. Therefore, $$\text{\DISJP}(X, Y) = 1 \iff G \text{ contains a cycle}.$$ Therefore by Theorem~\ref{thm:DISJ} it follows that $D^{1/400}_{\sigma_{C_1}}(\text{cycle-freeness}) = \Omega(h)$.

\begin{corollary}
\label{cor:cycle}
$R^{1/3}(\text{cycle-freeness without edge duplication}) = \Omega(\min\{m,n\})$.
\end{corollary}

\subsubsection{With Edge Duplication}
For the lower bound, we reduce from \threshP. For each $j \in [n]$, $G$ contains a vertex $v_j$. $G$ also contains a special vertex $u$. The total number of vertices in $G$ is $n+1$. 

Given an input $\{X_1, \ldots, X_k\} \sim \zeta$ for \threshP, the $k$ sites create a graph $G$ for cycle-freeness as follows. Each $P_i$ creates an edge $(u, v_j)$ for each $X_{i,j} = 1$. In addition, $P_1$ reconstructs $Y$, picks an arbitrary set of $\bar{Y} \subset [n] \backslash Y$ of size $(n+1)/4$, and creates an arbitrary perfect matching between $Y$ and $\bar{Y}$. Let $\sigma_{C_2}$ be the resulting input distribution of $G$. By Lemma~\ref{lem:noise}, we have with probability $(1 - 1/k^{10})$ that all pairs $(u, v_j)$ for $j \in [n] \backslash Y$ are connected. It is easy to see from the reduction that if $\text{\threshP}(X_1, \ldots, X_k) = 1$, then there is a cycle of the form $s \to v_i \to v_j \to s$ for some $i \in [n] \backslash Y$ and $j \in Y$. Otherwise the graph is a forest. Therefore, 
$$\text{\threshP}(X_1, \ldots, X_k) = 0 \iff  G \text{ is cycle-free}.$$
Thus by Theorem~\ref{thm:thresh}, we have that $D^{1/k^4}_{\sigma_{C_2}}(\text{cycle-freeness}) = \Omega(kn)$.
\begin{corollary}
\label{cor:cycle-2}
$R^{1/3}(\text{cycle-freeness with edge duplication}) = \Omega(k n/\log k)$.
\end{corollary}

There is again a trivial $O(kn\log n)$ upper bound. Each site first checks its local graph and reports directly if a cycle is found, otherwise the site sends all its edges (there are no more than $n-1$ such edges for a cycle-free graph) to $P_1$. Finally $P_1$ checks the cycle-freeness on the union of those edges.

\subsection{Connectivity and \#CC}
\label{sec:conn}
In the connectivity problem, the $k$ sites want to check whether $G$ is connected. In the \#connected-components (\#CC) problem, the $k$ sites want to compute the number of connected components in $G$. Note that solving \#CC also solves connectivity, thus we only show the lower bound for connectivity. 

\subsubsection{Without Edge Duplication}
For the lower bound, we reduce from \threshM. For each $i \in [k]$, $G$ contains a vertex $u_i$; and for each $j \in [r]$, $G$ contains a vertex $v_j$. The total number of vertices in $G$ is $n + k \le 2n$. Given an input $\{X_1, \ldots, X_k\} \sim \zeta$ for \threshR, the $k$ sites create a graph $G$ for connectivity as follows. Each $P_i$ creates an edge $(u_i, v_j)$ for each $X_{i,j} = 1$. In addition, $P_1$ reconstructs $Y$, and then creates a path containing $\{v_j\ |\ j \in Y\}$ and a path containing $\{v_j\ |\ j \in [r] \backslash Y\}$. See Figure~\ref{fig:conn} for an illustration. Let $\sigma_{N_1}$ be the resulting input distribution of $G$. It is easy to see from the reduction that 
$$\text{\threshR}(X_1, \ldots, X_k) = 1 \iff  G \text{ is connected}.$$
Thus by Theorem~\ref{thm:thresh}, we have that $D^{1/k^4}_{\sigma_{N_1}}(\text{connectivity}) = \Omega(kr)$.

\begin{corollary}
\label{cor:conn}
$R^{1/3}(\text{connectivity without edge duplication}) = \Omega(k r/\log k)$.
\end{corollary}

\begin{figure}[t]
\centering
\includegraphics[height = 2.0in]{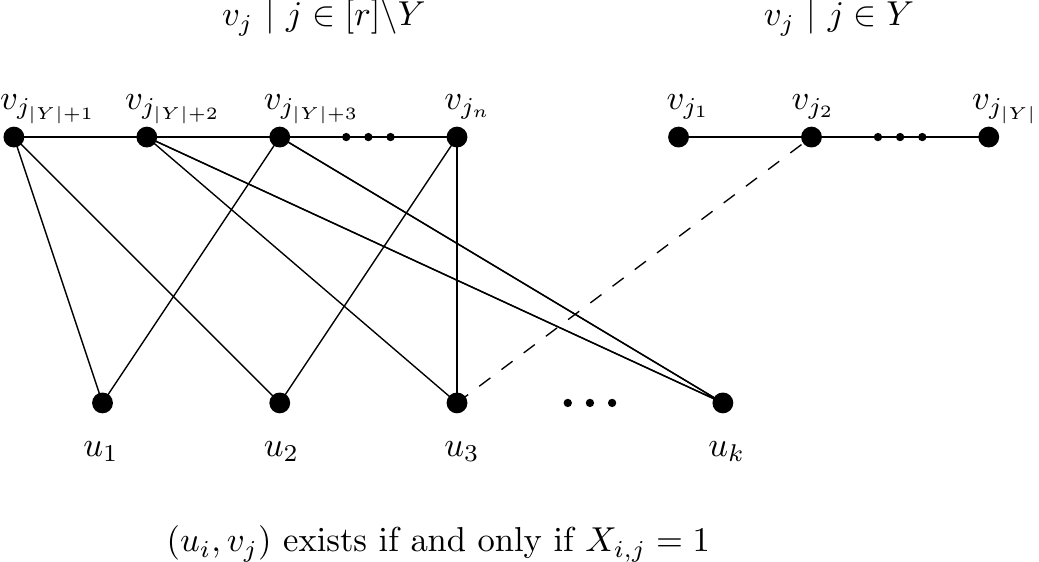}
\caption{Graph $G$ in the reduction for connectivity.}
\label{fig:conn}
\end{figure}

For the upper bound, all connected components (thus \#CC and connectivity) can be found by the protocol in which all sites send their local spanning trees to the first site $P_1$ and then $P_1$ does a local computation, which costs $O(kr \log n)$ bits of communication. 

\subsubsection{With Edge Duplication}
For the lower bound, we use a slightly modified reduction of the one for the without edge duplication case. We reduce from \threshP. For each $j \in [n]$, $G$ contains a vertex $v_j$. $G$ also contains a special vertex $u$. The total number of vertices in $G$ is $n + 1$. Given an input $\{X_1, \ldots, X_k\} \sim \zeta$ for \threshP, the $k$ sites create a graph $G$ for connectivity as follows. Each $P_i$ creates an edge $(u, v_j)$ for each $X_{i,j} = 1$. In addition, $P_1$ reconstructs $Y$, and then creates a path containing $\{v_j\ |\ j \in Y\}$ and a path containing $\{v_j\ |\ j \in [n] \backslash Y\}$. The total number of edges is $O(n) = O(m)$. This graph can be seen as the graph we constructed for the without edge duplication case after merging $u_1, \ldots, u_k$ to a single vertex $u$ while maintaining all the adjacent edges. Let $\sigma_{N_2}$ be the resulting input distribution of $G$. The correctness of the reduction is the same as before, that is, $\text{\threshP}(X_1, \ldots, X_k) = 1$ if and only if  $G \text{ is connected}$.
Thus by Theorem~\ref{thm:thresh}, we have that $D^{1/k^4}_{\sigma_{N_2}}(\text{connectivity}) = \Omega(kn)$.

\begin{corollary}
\label{cor:conn-2}
$R^{1/3}(\text{connectivity with edge duplication}) = \Omega(k n/\log k)$.
\end{corollary}

The upper bound is the same as the without edge duplication case, and the cost is $O(kn\log n)$ bits (note that since we allow edge duplication here, the total number of edges of the $k$ spanning trees cannot be bounded by $O(m)$).

\subsection{Bipartiteness}
\label{sec:bipartiteness}
In the bipartiteness problem, the $k$ sites want to check whether $G$ is bipartite. 

\subsubsection{Without Edge Duplication}
\label{sec:bipartite-without}
For the lower Bound, we reduce from \threshM. Given an input $\{X_1, \ldots, X_k\} \sim \zeta$ for \threshR, the $k$ sites create a graph $G = (V, E)$ with $V = A \cup B \cup C$ where $A = \{a_1, \ldots, a_r\}$, $B = \{b_1, \ldots, b_r\}$ and $C = \{c_1, \ldots, c_k\}$, as follows: each $P_i$ creates an edge $(c_i, b_j)$ for each $X_{i,j} = 1$. In addition, $P_1$ does the following:
\begin{enumerate}
\item Creates an edge $(a_i, b_i)$ for each $i \in [r]$.
\item Reconstructs $Y$. For each $i \in [k]$ and $j \in Y$, creates an edge $(c_i, a_j)$.
\end{enumerate}
The total number of vertices of $G$ is $2r+k < 3n$.
Let $\sigma_{B_1}$ be the resulting input distribution of $G$. One can see from the reduction that if $\text{\threshR}(X_1, \ldots, X_k) = 1$, then there exists at least one triangle in the form of $(a_i, b_i, c_j)$ for some $i \in [r], j \in [k]$. Otherwise if 
$$\text{\threshR}(X_1, \ldots, X_k) = 0,$$ then all edges are between two vertex sets $\{a_i\ |\ i \in [r] \backslash Y\} \cup  C \cup \{b_i\ |\ i \in Y\}$ and $\{b_i\ |\ i \in [r] \backslash Y\} \cup \{a_i\ |\ i \in Y\}$ and consequently $G$ is a bipartite graph. Therefore, $$\text{\threshR}(X_1, \ldots, X_k) = 0 \iff  G \text{ is bipartite}.$$
Thus by Theorem~\ref{thm:thresh}, we have that $D^{1/k^4}_{\sigma_{B_1}}(\text{bipartiteness}) = \Omega(kr)$.

\begin{corollary}
\label{cor:bipartiteness}
$R^{1/3}(\text{bipartiteness without edge duplication}) = \Omega(k r/\log k)$.
\end{corollary}

For the upper bound, we can assume that the graph is connected, since otherwise we can first compute all connected components (which costs $O(kn \log n)$ bits of communication as mentioned in Section~\ref{sec:conn}), and then work on each connected component. The protocol works as follows: the first site $P_1$ chooses an arbitrary vertex $u$ in the graph, and grows a breadth-first-search (BFS) tree rooted at $u$ by communicating with the other $k-1$ sites. In the first round, $P_1$ asks each site to report the vertices adjacent to $u$ using its local edges. The communication is at most $O(|N(u)| \log n \cdot k)$ bits, where $N(u)$ denotes the set of neighbors of $u$, and $|N(u)|$ denotes the number of (distinct) neighbors of $u$. From this, $P_1$ computes the entire set $N(u)$ of neighbors of $u$, without duplication, and sends it to the other $k-1$ sites. This also takes $O(|N(u)| \log n \cdot k)$ bits of communication. Now the sites all know $N(u)$, and they can build the first layer of the BFS tree rooted at $u$. Next, $P_1$ picks the first child $v$ (according to an arbitrary but fixed order) of $u$, and repeats this process on $v$. If ever a site finds an odd cycle, it is announced to $P_1$. Notice that every vertex is sent at most $k$ times to $P_1$, meanwhile the total number of vertices sent is no more than $O(m)$, so the total communication is $O(kr\log n)$ bits. 

\begin{remark}
We notice that if a graph is node-partitioned among the $k$ sites, that is, each node is stored at one site together with all its adjacent edges, then we can directly implement the algorithm for connectivity in \cite{AGM12a} in the message-passing model using $\tilde{O}(k + n)$ bits of communication. This shows a sharp difference between node-partition and edge-partition for connectivity with respect to the input storage.
\end{remark}

\subsubsection{With Edge Duplication}
We reduce from \threshP. The reduction is a simple modification of the one for the without edge duplication case. Given an input $\{X_1, \ldots, X_k\} \sim \zeta$ for \threshP, the $k$ sites create a graph $G = (V, E)$ with $V = A \cup B \cup C$ where $A = \{a_1, \ldots, a_n\}$, $B = \{b_1, \ldots, b_n\}$ and $C = \{c\}$, as follows: each $P_i$ creates an edge $(c, b_j)$ for each $X_{i,j} = 1$. In addition, $P_1$ creates an edge $(a_i, b_i)$ for each $i \in [n]$, reconstructs $Y$, and for each $j \in Y$ creates an edge $(c, a_j)$.
The total number of vertices of $G$ is $2n+1$.
Let $\sigma_{B_2}$ be the resulting input distribution of $G$. The correctness of the reduction is similar as before, that is, $\text{\threshP}(X_1, \ldots, X_k) = 0$ if and only if $G \text{ is bipartite}$. Thus by Theorem~\ref{thm:thresh}, we have that $D^{1/k^4}_{\sigma_{B_2}}(\text{bipartiteness}) = \Omega(kn)$.

\begin{corollary}
\label{cor:bipartiteness-2}
$R^{1/3}(\text{bipartiteness with edge duplication}) = \Omega(kn/\log k)$.
\end{corollary}

The upper bound is the same as the without edge duplication case, and the communication complexity is $O(kn\log n)$ bits. Note that since we have edge duplication here, the claim that ``the total number of vertices sent is no more than $O(m)$" does not hold.

\subsection{Triangle-freeness}
In the triangle-freeness problem, the $k$ sites want to check whether $G$ contains a triangle.

\subsubsection{Without Edge Duplication}
\label{sec:triangle-without}
An $O(m \log n)$ upper bound is the following: $P_2, \ldots, P_k$  send all their edges to $P_1$ and then $P_1$ does a local check.

There is an $\Omega(m)$ bit lower bound on the communication which holds even when $k = 2$, by a reduction from \DISJE. Suppose $P_1$ holds $X$ and $P_2$ holds $Y$, where $\{X, Y\} \sim \tau_{1/4}$ is the hard input distribution for \DISJE. Sites $P_1$ and $P_2$ construct a graph $G = (V, E)$ with $V = A \cup B \cup C$ where $A = \{a_1, \ldots, a_n\}$, $B = \{b_1, \ldots, b_n\}$ and  $C = \{c_1, \ldots, c_n\}$ as follows: for each $i \in X$, $P_1$ creates an edge $(a_p, c_q)$ such that $(p-1)n + q = i\ (p, q \in [n])$ (note that the solution of $(p, q)$ is unique). He/she also creates an edge $(a_t, b_t)$ for all $t \in [n]$. Similarly, for each $i \in Y$, $P_2$ creates an edge $(b_p, c_q)$ such that $(p-1)n + q = i\ (p, q \in [n])$. The graph $G$ has $3n$ vertices and $O(m)$ edges. See Figure~\ref{fig:triangle} for an illustration. Let $\sigma_{T_1}$ be the resulting input distribution of $G$.
One can see from the reduction that if $X \cap Y \neq  \emptyset$, then there is a triangle in the form of $(a_p, b_p, c_q)$ for some $p, q \in [n]$. Otherwise the graph is triangle-free. Therefore,
$$\text{\DISJE}(X, Y) = 0 \iff G \text{ is triangle free}.$$ Therefore by Theorem~\ref{thm:DISJ} we have $D^{1/400}_{\sigma_{T_1}}(\text{triangle-freeness}) = \Omega(m)$.

\begin{corollary}
\label{cor:triangle}
$R^{1/3}(\text{triangle-freeness without edge duplication}) = \Omega(m)$.
\end{corollary}

\begin{figure}[t]
\centering
\includegraphics[height = 2.4in]{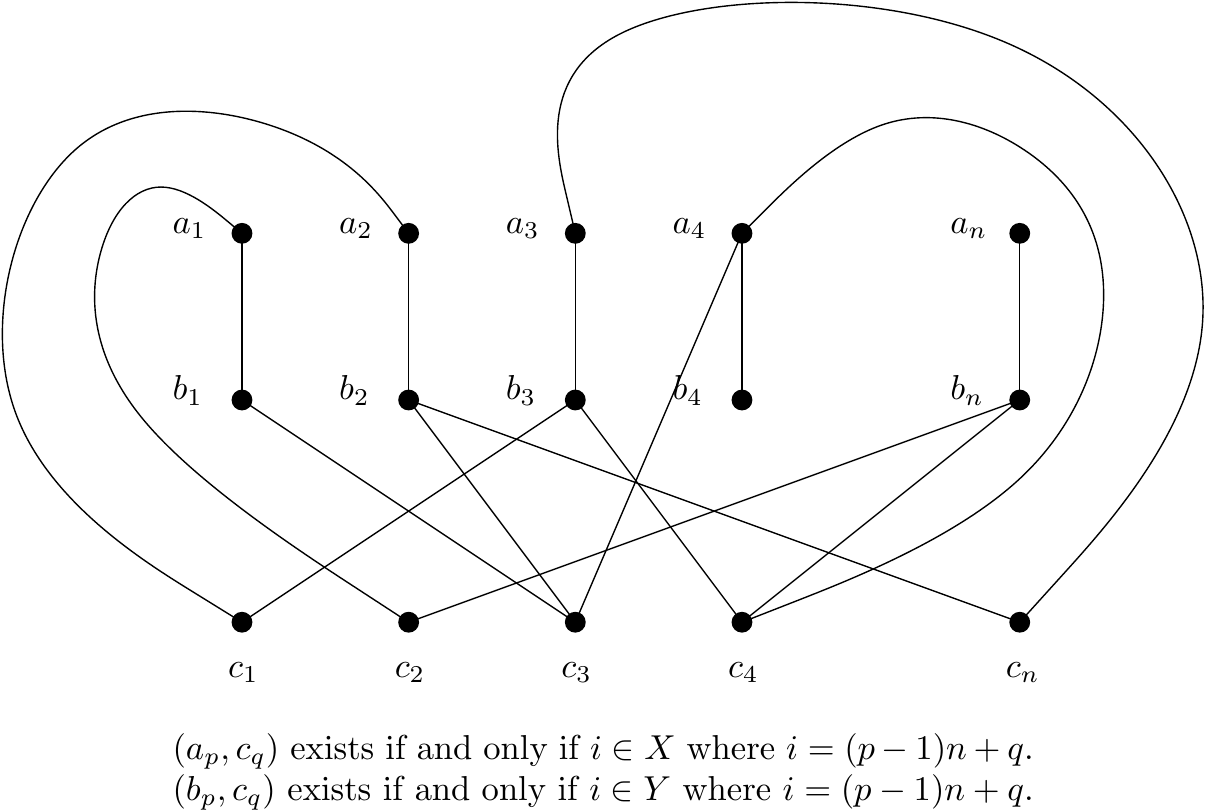}
\caption{Graph $G$ in the reduction for triangle-freeness.}
\label{fig:triangle}
\end{figure}

\subsubsection{With Edge Duplication}
We reduce from \threshE, and prove an $\Omega(km)$ lower bound on the communication cost. The reduction is an extension of the one for the without edge duplication case.

Given an input $\{X_1, \ldots, X_k\} \sim \zeta$ for \threshE, the $k$ sites create the following input graph $G = (V, E)$ for triangle-freeness with $V = A \cup B \cup C$ where $A = \{a_1, \ldots, a_n\}$, $B = \{b_1, \ldots, b_n\}$ and  $C = \{c_1, \ldots, c_n\}$. Each site $P_i$ does the following: for each $j \in [m]$ such that $X_{i,j} = 1$, the site creates an edge $(a_p, c_q)$ such that $(p-1)n + q = j\ (p, q \in [n])$. In addition, the first site $P_1$ also does the following.
\begin{enumerate}
\item
Creates an edge $(a_t, b_t)$ for each $t \in [n]$.
\item Reconstructs $Y$. For each $j \in Y$, creates an edge $(b_p, c_q)$ such that $(p-1)n + q = j\ (p, q \in [n])$.
\end{enumerate}
Let $\sigma_{T_2}$ be the resulting input distribution of $G$.
As before, it is easy to see that if $\text{\threshE}(X_1, \ldots, X_k) = 1$, then there is a triangle of the form $(a_p, b_p, c_q)$ for some $p, q \in [n]$. Otherwise the graph is triangle-free. Therefore,
$$\text{\threshE}(X_1, \ldots, X_k) = 0 \iff  G \text{ is triangle-free}.$$
By Theorem~\ref{thm:thresh}, we have that $D^{1/k^4}_{\sigma_{T_2}}(\text{triangle-freeness}) = \Omega(km)$.

\begin{corollary}
\label{cor:triangle-2}
$R^{1/3}(\text{triangle-freeness with edge duplication}) = \Omega(k m/\log k)$.
\end{corollary}

There is a simple protocol with $O(km \log n)$ bits of communication: each site sends all its edges to $P_1$ and the $P_1$ does a local check. 

We comment that our upper and lower bounds also applies to testing clique-freeness, that is, the $k$ sites want to check whether $G$ contains a clique of size $s$ for a fixed constant $s$.

\subsection{A Conjecture on the Diameter Problem and an Approximation Algorithm}
\label{sec:diameter}
We would like to mention the diameter problem which cannot be solved by the technique introduced in this paper. In the diameter problem, the $k$ sites want to compute the diameter of a graph $G = (V, E)$ in which the edges are distributed amongst the $k$ sites. We conjecture the following:

\begin{conjecture}
\label{con:diameter}
The randomized communication complexity of the diameter problem in the message-passing model is $\tilde{\Omega}(k m)$ bits, assuming edge duplication is allowed.
\end{conjecture}

Note that the naive algorithm in which every site sends all of its edges to the first site will match this lower bound up to a logarithmic factor.

In \cite{DHZ00} an algorithm for constructing a graph spanner in the RAM model with an additive distortion $2$ is proposed. A graph spanner with an additive distortion $d$ preserves all pairwise distances of vertices in the original graph up to an additive error $d$. In particular, the diameter is preserved up to an additive error $d$ when saying all pairwise distances of vertices are preserved up to an additive error $d$. This algorithm can be easily implemented in the message-passing model with a communication complexity of $O(k n^{3/2} \log^2 n)$ bits.
%~\footnote{We can further reduce the communication complexity to $O(k n^{4/3} \log n)$ bits by allowing an additive error $6$ using an alternative spanner construction algorithm by \cite{BKMP10}}. 
This fact gives us a hint that in order to prove an $\tilde{\Omega}(k n^2)$ (in the case when $m = \Theta(n^2)$) lower bound, we have to explore the difficulty of distinguishing a diameter $T$ verses $T+1$ for a value $T \in [1, n - 2]$ (in fact, one can show that $T$ also needs to be a fixed constant). Now we briefly describe how to implement the algorithm in \cite{DHZ00} in the message-passing model. 

The algorithm in \cite{DHZ00} works as follows (in the RAM model).
%\vspace{0.1cm}

%\noindent{\bf Algorithm for constructing a graph spanner with an additive distortion $2$}
\begin{enumerate}
\item We pick $\Theta(\sqrt{n} \log n)$ vertices in the graph uniformly at random, and grow a breadth-first-search (BFS) tree rooted on each of these vertices. We then include all edges of these BFS trees into the spanner.

\item We include all edges incident to vertices whose degrees are no more than $\sqrt{n}$.
\end{enumerate}
In \cite{DHZ00} it was shown that this algorithm computes a spanner with an additive distortion $2$ correctly with probability $0.99$. Now let us briefly discuss how to implement this algorithm in the message-passing model. For the first step, the random sampling can be done by $P_1$ locally, and then $P_1$ communicates with the other $k - 1$ sites to grow a BFS tree rooted on each of these vertices using the algorithm described in Section~\ref{sec:bipartiteness}. Recall that constructing each of these BFS trees costs $O(kn\log n)$ bits of communication. Thus the total communication needed for the first step is $O(kn^{3/2} \log^2 n)$ bits. For the second step, the $k$ sites first use the algorithm in \cite{CMY11} to compute the degree of each vertex (this is essentially $F_0$) up to a factor of $2$, using $O(kn\log n)$ bits of communication, and then they construct the set $H = \{v \in V\ |\ \text{degree}(v) \le 2\sqrt{n}\}$. Next, $P_2, \ldots, P_k$ send $P_1$ all edges that are incident to a vertex in $H$. Note that each of $P_i\ (i = 2, \ldots, k)$ will send at most $O(n^{3/2})$ edges to $P_1$. Therefore, the total communication cost of the second step is bounded by $O(kn \log n + kn^{3/2} \log n) = O(kn^{3/2} \log n)$ bits. We conclude with the following theorem for diameter.
\begin{theorem}
There exists a randomized protocol that approximates the diameter of a graph up to an additive error of $2$ in the message-passing model. The protocol succeeds with probability $0.99$ and uses $O(kn^{3/2} \log^2 n)$ bits of communication.
\end{theorem}

%\subsection{Diameter}

\section{Concluding Remarks}
\label{sec:discussion}

In this paper we show that exact computation of many basic statistical and graph problems in the message-passing model are necessarily communication-inefficient. An important message we want to deliver through these negative results, which is also the main motivation of this paper, is that a relaxation of the problem, such as an approximation, is necessary in the distributed setting if we want communication-efficient protocols. Besides approximation, the layout and the distribution of the input are also important factors for reducing communication.

An interesting future direction is to further investigate efficient communication protocols for approximately computing statistical and graph problems in the message-passing model, and to explore realistic distributions and layouts of the inputs. 
%, for example, computing the number of connected components and triangles in a graph.
% David: since we show 0 versus non-zero triangle, or 1 versus 2 connected components, these
% are essentially already hard to approximate

One question which we have not discussed in this paper but is important for practice, is whether we can obtain round-efficient protocols that (almost) match the lower bounds which hold even for round-inefficient protocols? Most simple protocols presented in this paper only need a constant number of rounds, except the ones for bipartiteness and (approximate) diameter, where we need to grow BFS trees which are inherently sequential (require $\Omega(\Delta)$ rounds where $\Delta$ is the diameter of the graph). Using the sketching algorithm in \cite{AGM12a}, we can obtain a $1$-round protocol for bipartiteness that uses $\tilde{O}(kn)$ bits of communication. We do not know whether a round-efficient protocol exists for the additive-$2$ approximate diameter problem that could (almost) match the $\tilde{O}(kn^{3/2})$ bits upper bound obtained by the round-inefficient protocol in Section~\ref{sec:diameter}.

\bibliographystyle{abbrv}
\bibliography{paper}

\begin{thebibliography}{10}

\bibitem{AGM12a}
K.~J. Ahn, S.~Guha, and A.~McGregor.
\newblock Analyzing graph structure via linear measurements.
\newblock In {\em Proceedings of the Twenty-Third Annual ACM-SIAM Symposium on
  Discrete Algorithms}, pages 459--467. SIAM, 2012.

\bibitem{AGM12b}
K.~J. Ahn, S.~Guha, and A.~McGregor.
\newblock Graph sketches: sparsification, spanners, and subgraphs.
\newblock In {\em Proc. ACM Symposium on Principles of Database Systems}, pages
  5--14, 2012.

\bibitem{ABC09}
C.~Arackaparambil, J.~Brody, and A.~Chakrabarti.
\newblock Functional monitoring without monotonicity.
\newblock In {\em Proc. International Colloquium on Automata, Languages, and
  Programming}, 2009.

\bibitem{BBFM12}
M.-F. Balcan, A.~Blum, S.~Fine, and Y.~Mansour.
\newblock Distributed learning, communication complexity and privacy.
\newblock {\em Journal of Machine Learning Research - Proceedings Track},
  23:26.1--26.22, 2012.

\bibitem{BKS13}
P.~Beame, P.~Koutris, and D.~Suciu.
\newblock Communication steps for parallel query processing.
\newblock In {\em PODS}, pages 273--284, 2013.

\bibitem{BHMPRS05}
P.~Brown, P.~J. Haas, J.~Myllymaki, H.~Pirahesh, B.~Reinwald, and Y.~Sismanis.
\newblock Toward automated large-scale information integration and discovery.
\newblock In {\em Data Management in a Connected World}, pages 161--180, 2005.

\bibitem{CMY11}
G.~Cormode, S.~Muthukrishnan, and K.~Yi.
\newblock Algorithms for distributed functional monitoring.
\newblock {\em ACM Transactions on Algorithms}, 7(2):21, 2011.

\bibitem{DHZ00}
D.~Dor, S.~Halperin, and U.~Zwick.
\newblock All-pairs almost shortest paths.
\newblock {\em SIAM J. Comput.}, 29(5):1740--1759, 2000.

\bibitem{ER60}
P.~Erdős and A.~Rényi.
\newblock On the evolution of random graphs.
\newblock In {\em Publication of the mathematical institute of the hungarian
  academy of sciences}, pages 17--61, 1960.

\bibitem{EVF06}
C.~Estan, G.~Varghese, and M.~Fisk.
\newblock Bitmap algorithms for counting active flows on high-speed links.
\newblock {\em IEEE/ACM Trans. Netw.}, 14(5):925--937, Oct. 2006.

\bibitem{flajolet2008}
P.~Flajolet, {\'E}.~Fusy, O.~Gandouet, and F.~Meunier.
\newblock Hyperloglog: the analysis of a near-optimal cardinality estimation
  algorithm.
\newblock {\em DMTCS Proceedings}, (1), 2008.

\bibitem{flajolet85:_probab}
P.~Flajolet and G.~N. Martin.
\newblock Probabilistic counting algorithms for data base applications.
\newblock {\em Journal of Computer and System Sciences}, 31(2):182--209, 1985.

\bibitem{GGR98}
O.~Goldreich, S.~Goldwasser, and D.~Ron.
\newblock Property testing and its connection to learning and approximation.
\newblock {\em Journal of the ACM}, 45(4):653--750, 1998.

\bibitem{GSZ11}
M.~T. Goodrich, N.~Sitchinava, and Q.~Zhang.
\newblock Sorting, searching, and simulation in the mapreduce framework.
\newblock In {\em Proc. International Symposium on Algorithms and Computation},
  pages 374--383, 2011.

\bibitem{HRVZ13}
Z.~Huang, B.~Radunovi\'c, M.~Vojnovi\'c, and Q.~Zhang.
\newblock The communication complexity of approximate maximum matching in
  distributed data.
\newblock {\em Manuscript}, 2013.
\newblock http://research.microsoft.com/apps/pubs/default.aspx?id=188946.

\bibitem{HYZ12}
Z.~Huang, K.~Yi, and Q.~Zhang.
\newblock Randomized algorithms for tracking distributed count, frequencies,
  and ranks.
\newblock In {\em Proc. ACM Symposium on Principles of Database Systems}, pages
  295--306, 2012.

\bibitem{DPSV12a}
H.~D. III, J.~M. Phillips, A.~Saha, and S.~Venkatasubramanian.
\newblock Efficient protocols for distributed classification and optimization.
\newblock In {\em ALT}, pages 154--168, 2012.

\bibitem{DPSV12b}
H.~D. III, J.~M. Phillips, A.~Saha, and S.~Venkatasubramanian.
\newblock Protocols for learning classifiers on distributed data.
\newblock {\em Journal of Machine Learning Research - Proceedings Track},
  22:282--290, 2012.

\bibitem{Kane10}
D.~M. Kane, J.~Nelson, and D.~P. Woodruff.
\newblock An optimal algorithm for the distinct elements problem.
\newblock In {\em Proc. ACM Symposium on Principles of Database Systems}, pages
  41--52, 2010.

\bibitem{KSV10}
H.~J. Karloff, S.~Suri, and S.~Vassilvitskii.
\newblock A model of computation for mapreduce.
\newblock In {\em Proc. ACM-SIAM Symposium on Discrete Algorithms}, pages
  938--948, 2010.

\bibitem{KS11}
P.~Koutris and D.~Suciu.
\newblock Parallel evaluation of conjunctive queries.
\newblock In {\em Proc. ACM Symposium on Principles of Database Systems}, pages
  223--234, 2011.

\bibitem{KN97}
E.~Kushilevitz and N.~Nisan.
\newblock {\em Communication Complexity}.
\newblock Cambridge University Press, 1997.

\bibitem{PGF02}
C.~R. Palmer, P.~B. Gibbons, and C.~Faloutsos.
\newblock Anf: a fast and scalable tool for data mining in massive graphs.
\newblock In {\em Proc. ACM SIGKDD International Conference on Knowledge
  Discovery and Data Mining}, pages 81--90, 2002.

\bibitem{PVZ12}
J.~M. Phillips, E.~Verbin, and Q.~Zhang.
\newblock Lower bounds for number-in-hand multiparty communication complexity,
  made easy.
\newblock In {\em Proc. ACM-SIAM Symposium on Discrete Algorithms}, 2012.

\bibitem{Raz90}
A.~A. Razborov.
\newblock On the distributional complexity of disjointness.
\newblock In {\em Proc. International Colloquium on Automata, Languages, and
  Programming}, 1990.

\bibitem{Valiant90}
L.~G. Valiant.
\newblock A bridging model for parallel computation.
\newblock {\em Communications of the ACM}, 33(8):103--111, 1990.

\bibitem{WZ12}
D.~P. Woodruff and Q.~Zhang.
\newblock Tight bounds for distributed functional monitoring.
\newblock In {\em Proc. ACM Symposium on Theory of Computing}, 2012.

\bibitem{yao77}
A.~C. Yao.
\newblock Probabilistic computations: Towards a unified measure of complexity.
\newblock In {\em Proc. IEEE Symposium on Foundations of Computer Science},
  1977.

\end{thebibliography}

\appendix

\section{Proof for Lemma~\ref{lem:Yao}}
\label{sec:proof-Yao}

\begin{proof}
The original proof is for two players, though this also holds for $k > 2$ players since for any distribution $\mu$, if $\Pi$ is a $\delta$-error protocol then for all possible 
inputs $x^1, \ldots, x^k$ to the $k$ players,
$$\Pr_{\textrm{random tapes of the players}}[\Pi(x^1,\ldots, x^k) = f(x^1, \ldots, x^k)] \geq 1-\delta,$$
which implies for any distribution $\mu$ on $(x^1, \ldots, x^k)$ that
$$\Pr_{\textrm{random tapes of the players}, (x^1,\ldots, x^k) \sim \mu}[\Pi(x^1, \ldots, x^k) 
= f(x^1, \ldots, x^k)] \geq 1-\delta,$$
which implies there is a fixing of the random tapes of the players
so that $$\Pr_{(x^1, \ldots, x^k) \sim \mu}[\Pi(x^1, \ldots, x^k) 
= f(x^1, \ldots, x^k)] \geq 1-\delta,$$
which implies $D^{\delta}_{\mu}(f)$ is at most $R^{\delta}(f)$. 
\end{proof}

\end{document}